\documentclass{llncs}
\usepackage{makeidx}  % allows for indexgeneration

\usepackage{amssymb}
\usepackage{amsfonts}
\usepackage{latexsym}

\usepackage{pstricks}
\usepackage{xkeyval}
\usepackage{pst-node}
\usepackage{pstricks-add}
\usepackage{pstricks}
\usepackage{pst-node}

\newcommand{\cut}[1]{}

\begin{document}
\mainmatter              % start of the contributions
\title{Symbolic model checking of tense logics on rational Kripke models}
\titlerunning{Rational Kripke models}  % abbreviated title (for running head)
\author{Wilmari Bekker\inst{1} \and Valentin Goranko\inst{2}}
\authorrunning{Bekker and Goranko}   % abbreviated author list (for running head)
\tocauthor{Wilmari Bekker, Valentin Goranko}
\institute{University of Johannesburg and University of the Witwatersrand, Johannesburg,\\
\email{bekkerw@gmail.com}
\and
University of the Witwatersrand, Johannesburg,\\
\email{goranko@maths.wits.ac.za}}

\maketitle              % typeset the title of the contribution

\begin{abstract}
We introduce the class of \emph{rational Kripke models} and study symbolic model checking of the basic tense logic \textbf{K$_t$} and some extensions of it in models from that class. Rational Kripke models are based on (generally infinite) \emph{rational graphs}, with vertices labeled by the words in some regular language and transitions recognized by asynchronous two-head finite automata, also known as \emph{rational transducers}. Every atomic proposition in a rational Kripke model is evaluated in a regular set of states. We show that every formula of \textbf{K$_t$} has an effectively computable regular extension in every rational Kripke model, and therefore local model checking and global model checking of \textbf{K$_t$} in rational Kripke models are decidable. These results are lifted to a number of extensions of \textbf{K$_t$}. We study and partly determine the complexity of the model checking procedures.
\end{abstract}
\section{Introduction}\label{intro}

Verification of models with infinite state spaces using algorithmic symbolic model checking techniques has been an increasingly active area of research over recent years. One very successful approach to infinite state verification is based on the representation of sets of states and transitions by means of automata. It is the basis of various automata-based techniques for model checking, e.g., of linear and branching-time temporal logics on finite transition systems \cite{Var96,KVW00}, regular model checking \cite{bouajjani00regular}, pushdown systems \cite{BouEspMal97,Walukiewicz00,SEA2001}, automatic structures \cite{KhoussainovNerode,BlumensathGraedel} etc. In most of the studied cases of infinite-state symbolic model checking (except for automatic structures), the logical languages are sufficiently expressive for various reachability
properties, but the classes of models are relatively restricted.

In this paper we study a large and natural class of \emph{rational Kripke models}, on which global model checking of the basic tense\footnote{We use the term `tense' rather than `temporal' to emphasize that the accessibility relation is not assumed transitive, as in a usual flow of time.} logic \textbf{K$_t$} (with forward and backward one-step modalities) and of some extensions thereof, are decidable. The language of \textbf{K$_t$} is
sufficient for expressing \emph{local properties}, i.e., those referring to a bounded width neighborhood of predecessors or successors of the current state.
In particular, pre-conditions and post-conditions are local, but not reachability properties. Kesten et al \cite{KestenEtAl01} have formulated the following minimal requirements for an \emph{assertional language $\mathcal{L}$ to be adequate for symbolic model checking}:

\begin{enumerate}

\item
\textsf{The property to be verified and the initial conditions (i.e., the set of initial states) should be
expressible in $\mathcal{L}$.}
\item \textsf{ $\mathcal{L}$ should be effectively closed under the boolean operations, and should possess an
algorithm for deciding equivalence of two assertions.}
\item \textsf{ There should exist an algorithm for constructing the predicate transformer \textit{pred}, where
\textit{pred}$(\phi)$ is an assertion characterizing the set of states that have a successor state satisfying
$\phi$.}
\end{enumerate}

Assuming that the property to be verified is expressible in \textbf{K$_t$}, the first condition above is satisfied in our case. Regarding the set of initial states, it is usually assumed a singleton, but certainly an effective set, and it can be represented by a special modal constant $S$. The second condition is clearly satisfied, assuming the equivalence is with respect to the model on which the verification is being done. As for the third condition, \textit{pred}$(\phi) = \left\langle R \right\rangle \phi$. Thus, the basic modal logic \textbf{K} is \emph{the minimal natural logical language satisfying these requirements}, and hence it suffices for specification of \emph{pre-conditions} over regular sets of states. The tense extension \textbf{K$_t$} enables specification of post-conditions, as well, thus being the basic adequate logic for specifying \emph{local properties} of transition systems and warranting the potential utility of the work done in the present paper. In particular, potential areas of applications of model checking of the basic tense logic to verification of infinite state systems are \emph{bounded model checking} \cite{biere2003}, applied to infinite state systems, and (when extended with reachability) \emph{regular model checking} \cite{bouajjani00regular} -- a framework for algorithmic verification of generally infinite state systems which essentially involves computing reachability sets in regular Kripke models.

The paper is organized as follows: in Section \ref{sec:prel} we introduce {\bf K$_t$} and rational transducers. Section \ref{sec:RKM} introduces and discusses rational Kripke models, and in Section \ref{sec:Synchronized} we introduce synchronized products of transducers and automata. We use them in Section \ref{sec:MCinRKM} to show decidability of global and local symbolic model checking of {\bf K$_t$} in rational Kripke models and in Section \ref{sec:Complexity} we discuss its complexity. The model checking results are strengthened in Section \ref{sec:OtherExt} to hybrid and other extensions of {\bf H$_t(U)$}, for which some model checking tasks remain decidable.

\section{Preliminaries}
\label{sec:prel}

\subsection{The basic tense logic {\bf K$_t$}}

We consider transition systems with one transition relation $R$. The \emph{basic tense logic} \textbf{K$_t$} for such transition systems extends the classical propositional logic with two unary modalities: one associated with $R$ and the other with its inverse $R^{-1}$, respectively denoted by $[R]$ and $[R^{-1}]$.  The generalization of what follows to the case of languages and models for transition systems with many relations is straightforward.  Note that the relation $R$ is not assumed transitive, and therefore the language of \textbf{K$_t$} cannot express $R$-reachability properties.

\subsection{Rational transducers and rational relations}

\textbf{Rational transducers}, studied by Eilenberg \cite{Eilenberg}, Elgot and Mezei \cite{ElgotMezei}, Nivat, Berstel \cite{Berstel}, etc., are \emph{asynchronous automata} on pairs of words. Intuitively, these are finite automata with two autonomous heads that read the input pair of words asynchronously, i.e. each of them can read arbitrarily farther ahead of the other. The transitions are determined by a finite set of pairs of (possibly empty) words; alternatively, a transition can be labeled either by a pair of
letters (when both heads make a move on their respective words) or by $\left\langle a,\epsilon \right\rangle $ or $\left\langle \epsilon , a \right\rangle $, where $a$ is a letter, and $\epsilon$ is the empty word (when one of the heads reads on, while the other is waiting). The formal definition follows.

\begin{definition}
A  \textbf{(rational) transducer} is a tuple
$\mathcal{T}=\left<Q,\Sigma,\Gamma,q_i,F,\rho\right>$ where $\Sigma$ and $\Gamma$ are the input and output alphabets respectively, $Q$ a set of states, $q_i\in Q$ a unique starting state, $F\subseteq Q$ a set of accepting states and $\rho \subseteq Q\times (\Sigma \cup \{\varepsilon \})\times (\Gamma \cup \{\varepsilon \})\times Q$ is the transition relation, consisting of \emph{finitely many} tuples, each containing the current state, the pair of letters (or $\varepsilon$) triggering the transition, and the new state. Alternatively, one can take $\rho \subseteq Q\times \Sigma ^{\ast }\times
\Gamma ^{\ast }\times Q$.

The  language recognized by the transducer $\mathcal{T}$ is the set of all pairs of words for which it has a reading that ends in an accepting state. Thus, the transducer $\mathcal{T}$ recognizes a binary relation $R\subseteq\Sigma^*\times\Gamma^*$.
\end{definition}

This is the `static' definition of rational transducers; they can also be defined `dynamically', as reading an input word, and transforming it into an output word, according to the transition relation which is now regarded as a mapping from words to sets of words (because it can be non-deterministic).

\begin{example}\label{binarytreerat1}
For $\mathcal{T}=\left<Q,\Sigma,\Gamma,q_i,F,\rho\right>$ let: $
Q=\left\{q_1,q_2\right\}; \ \Sigma = \left\{0,1\right\}=\Gamma; \ q_i=q_1; \ F =
\left\{q_2\right\}; \rho =
\left\{\left(q_1,0,0,q_1\right),\left(q_1,1,1,q_1\right),\left(q_1,\epsilon,0,q_2\right),
\left(q_1,\epsilon,1,q_2\right)\right\} $

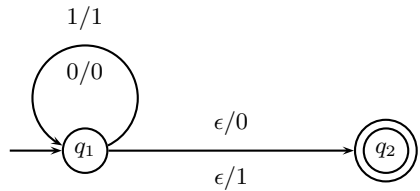
\begin{figure}[here]
\caption{The transducer $\mathcal{T}$ which recognizes pairs of words of the forms
$\left(u,u0\right)$ or $\left(u,u1\right)$ where $u\in\Sigma^*$}

\begin{center}
\begin{pspicture}(8,4)
\pnode(1,1){p0}
\cnodeput(2,1){q1}{$q_1$}
\cnodeput(6,1){q2}{$q_2$}
\cnode(6,1){0.41}{qf}
\ncline{->}{q1}{qf}
\naput{$\epsilon /0$}
\nbput{$\epsilon /1$}
\nccircle{->}{q1}{0.7cm}
\naput{$0/0$} \nbput{$1/1$}
\ncline{->}{p0}{q1}
\end{pspicture}
\end{center}
\end{figure}

Notice that in the representation of $\mathcal{T}$ there is only one edge between two states but that an edge may have more than one label.

\end{example}

A relation $R\subseteq\Sigma^*\times\Gamma^*$ is \textbf{rational} if it is recognizable by a rational transducer. Equivalently (see \cite{Berstel}), given finite alphabets $\Sigma,\Gamma$, a (binary) \textbf{rational relation} over $(\Sigma,\Gamma)$ is a rational subset of $\Sigma^* \times \Gamma^*$, i.e., a subset generated by a rational expression (built up using union, concatenation, and iteration) over a finite subset of $\Sigma^* \times \Gamma^*$. Hereafter, we will assume that the input and output alphabets $\Sigma$ and $\Gamma$
coincide.

Besides the references above, rational relations have also been studied by Johnson \cite{Johnson86}, Frougny and Sakarovich \cite{FrougnySakarovitch93}, and more recently by Morvan \cite{morvan}. It is important to note that the class of rational relations is closed under \emph{unions}, \emph{compositions}, and \emph{inverses} \cite{Berstel}. On the other hand, the class of rational relations is not closed under intersections, complements, and transitive closure (\emph{ibid}).

\section{Rational Kripke models}
\label{sec:RKM}

\subsection{Rational graphs}

\begin{definition}
A graph $\mathcal{G}=\left(S,E\right)$ is \textbf{rational}, if the set of vertices $S$ is a regular language in some finite alphabet $\Sigma$ and the set of edges $E$ is a rational relation on $\Sigma$.
\end{definition}

\begin{example}
\textbf{The infinite grid.} Let $\Sigma=\left\{0,1\right\}$, then the infinite grid with vertices in $\Sigma^*$ is given by Figure \ref{infgrid} and the edge relation of this graph is recognized by the transducer given in Figure \ref{infgrid}.

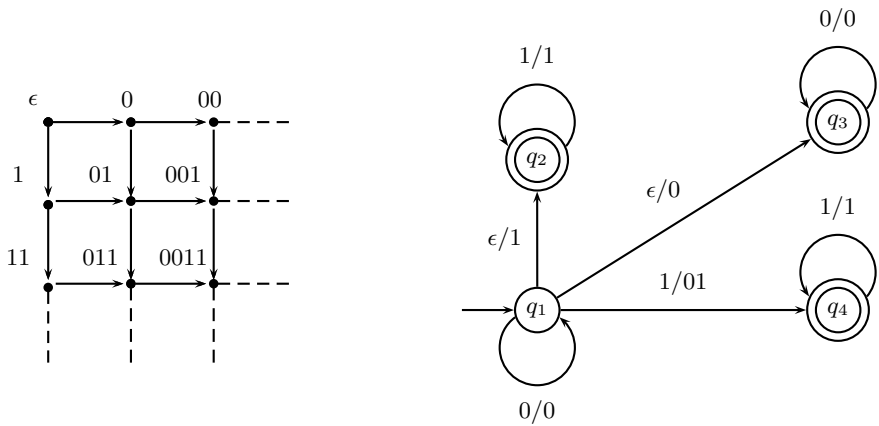
\begin{figure}[here]
\caption{The infinite grid with set of vertices $S=0^*1^*$ and a transducer that recognizes the infinite grid.}
\label{infgrid}
\vspace{1cm}

\begin{pspicture}(4.5,4.5)
\rput(0.3,3.8){\rnode{e}{$\epsilon$}}
\rput(1.55,3.8){\rnode{e}{0}} \rput(2.65,3.8){\rnode{e}{00}}
\psline[border=2pt]{*->}(0.5,3.5)(1.5,3.5)
\psline[border=2pt]{*->}(1.6,3.5)(2.6,3.5)
\psline[linestyle=dashed]{*-}(2.7,3.5)(3.7,3.5)
\rput(0.1,2.8){\rnode{e}{1}} \rput(0.1,1.7){\rnode{e}{11}}
\psline[border=2pt]{*->}(0.5,3.5)(0.5,2.5)
\psline[border=2pt]{*->}(0.5,2.4)(0.5,1.4)
\psline[linestyle=dashed]{*-}(0.5,1.3)(0.5,0.3)
\psline[border=2pt]{->}(0.6,2.45)(1.5,2.45)
\psline[border=2pt]{*->}(1.6,2.45)(2.6,2.45)
\psline[linestyle=dashed]{*-}(2.7,2.45)(3.7,2.45)
\psline[border=2pt]{->}(0.6,1.35)(1.5,1.35)
\psline[border=2pt]{*->}(1.6,1.35)(2.6,1.35)
\psline[linestyle=dashed]{*-}(2.7,1.35)(3.7,1.35)
\psline[border=2pt]{->}(1.6,3.4)(1.6,2.5)
\psline[border=2pt]{->}(1.6,2.35)(1.6,1.4)
\psline[linestyle=dashed]{-}(1.6,1.3)(1.6,0.3)
\psline[border=2pt]{->}(2.7,3.4)(2.7,2.5)
\psline[border=2pt]{->}(2.7,2.35)(2.7,1.4)
\psline[linestyle=dashed]{-}(2.7,1.3)(2.7,0.3)
\rput(1.2,2.8){\rnode{e}{01}} \rput(2.3,2.8){\rnode{e}{001}}
\rput(1.2,1.7){\rnode{e}{011}} \rput(2.3,1.7){\rnode{e}{0011}}

\pnode(6,1){p0} \cnodeput(7,1){q1}{$q_1$}
\cnodeput(7,3){q2}{$q_2$} \cnode(7,3){0.41}{2}
\cnodeput(11,3.5){q3}{$q_3$} \cnode(11,3.5){0.41}{3}
\cnodeput(11,1){q4}{$q_4$} \cnode(11,1){0.41}{4} \ncline{->}{p0}{q1}
\ncline{->}{q1}{2} \naput{$\epsilon/1$} \ncline{->}{q1}{3}
\naput{$\epsilon/0$} \ncline{->}{q1}{4} \naput{$1/01$}
\nccircle[angleA=180]{->}{q1}{0.5cm} \nbput{$0/0$}
\nccircle{->}{2}{0.5cm} \nbput{$1/1$} \nccircle{->}{3}{0.5cm}
\nbput{$0/0$} \nccircle{->}{4}{0.5cm} \nbput{$1/1$}
\end{pspicture}

\end{figure}

\end{example}

\bigskip

\begin{example} \label{bintreeex}
The \textbf{complete binary tree} $\Lambda$.

Figure \ref{combintree} contains the complete binary
tree with vertices in $\left\{0,1\right\}^*$ and labeled by $\Gamma=\left\{a,b\right\}$, as well as
the transducer recognizing it, in which the accepting states are labeled
respectively by $a$ and $b$. The pairs of words for which the transducer ends in the accepting state $q_4$ belong to the left successor relation in the tree (labeled by $a$), and those for which the transducer ends in the accepting state $q_5$ belong to the right successor relation in the tree (labeled by $b$).

\begin{figure}[here]
\caption{The complete binary tree $\Lambda$ and a  labeled transducer recognizing it.} \label{combintree}
\vspace{2cm}

\psset{unit=0.9}
\begin{pspicture}(5,3.5)
\rput(2.5,3.2){\rnode{e}{$\epsilon$}}
\psline[border=2pt]{->}(2.5,3)(1.5,2)
\rput(1.9,2.7){\rnode{e}{$a$}}
\psline[border=2pt]{*->}(2.5,3)(3.5,2)
\rput(3.1,2.7){\rnode{e}{$b$}}
\rput(1.3,2){\rnode{e}{0}} \rput(3.7,2){\rnode{e}{1}}
\psline[border=2pt]{->}(1.5,2)(0.8,0.8)
\rput(0.9,1.4){\rnode{e}{$a$}}
\psline[border=2pt]{*->}(1.5,2)(2.2,0.8)
\rput(2.1,1.4){\rnode{e}{$b$}}
\psline[linestyle=dashed]{*-}(0.8,0.8)(0.4,0)
\psline[linestyle=dashed]{-}(0.8,0.8)(1.2,0)
\psline[linestyle=dashed]{*-}(2.2,0.8)(1.8,0)
\psline[linestyle=dashed]{-}(2.2,0.8)(2.4,0)
\rput(0.5,0.9){\rnode{e}{00}} \rput(1.8,0.9){\rnode{e}{01}}
\psline[border=2pt]{->}(3.5,2)(2.8,0.8)
\rput(2.9,1.4){\rnode{e}{$a$}}
\psline[border=2pt]{*->}(3.5,2)(4.2,0.8)
\rput(4.1,1.4){\rnode{e}{$b$}}
\psline[linestyle=dashed]{*-}(2.8,0.8)(2.6,0)
\psline[linestyle=dashed]{-}(2.8,0.8)(3.2,0)
\psline[linestyle=dashed]{*-}(4.2,0.8)(3.8,0)
\psline[linestyle=dashed]{-}(4.2,0.8)(4.6,0)
\rput(3.2,0.9){\rnode{e}{10}} \rput(4.4,0.9){\rnode{e}{11}}

\pnode(9.5,0){p0} \cnodeput(9.5,1){q1}{$q_1$}
\cnodeput(6,2){q2}{$q_2$} \cnodeput(13,2){q3}{$q_3$}
\cnodeput(7,4){q4}{$q_4$} \cnodeput(12,4){q5}{$q_5$}
\cnode(7,4){0.41}{f1} \cnode(12,4){0.41}{f2}
\nccircle[angleA=180]{->}{q2}{0.6cm} \naput{$0/0$} \nbput{$1/1$}
\nccircle[angleA=180]{->}{q3}{0.6cm} \naput{$0/0$} \nbput{$1/1$}
\ncline{->}{p0}{q1} \ncline{->}{q1}{q2} \naput{$0/0$}
\ncline{->}{q1}{q3} \nbput{$1/1$} \ncline{->}{q2}{f2}
\naput[npos=0.9]{$\epsilon/1$} \ncline{->}{q2}{f1}
\naput{$\epsilon/0$} \ncline{->}{q3}{f1}
\nbput[npos=0.9]{$\epsilon/0$} \ncline{->}{q3}{f2}
\nbput{$\epsilon/1$} \ncline{->}{q1}{f1}
\naput[npos=0.9]{$\epsilon/0$} \ncline{->}{q1}{f2}
\nbput[npos=0.9]{$\epsilon/1$} \rput(6.2,4.5){\rnode{a}{$a$}}
\rput(12.8,4.5){\rnode{b}{$b$}}
\end{pspicture}

\end{figure}
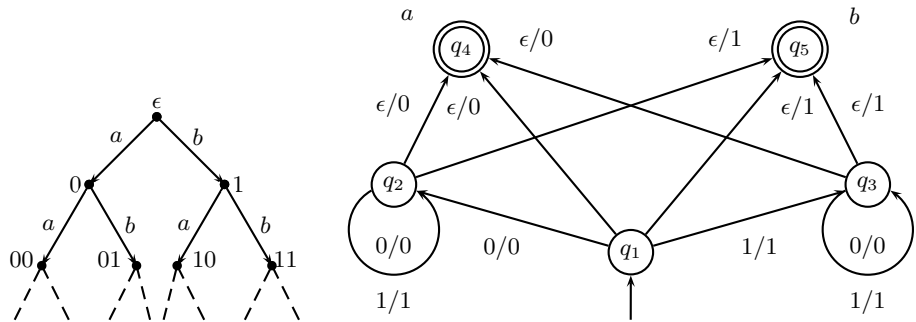
\end{example}

An important and extensively studied subclass of rational graphs is the class of \emph{automatic graphs}  \cite{KhoussainovNerode,BlumensathGraedel}.  These are rational graphs whose transition relations are recognized by \emph{synchronized} transducers.

As shown by Blumensath \cite{BlumensathThesis}, the configuration graph of every Turing machine is an automatic graph. Consequently, important queries, such as reachability, are generally undecidable on automatic graphs, and hence on rational graphs. Furthermore, Morvan showed in \cite{morvan} that the configuration graphs of Petri nets \cite{Reis1985} are rational (in fact, automatic) graphs, too.

Moreover, Johnson \cite{Johnson86} proved that even very simple first-order definable properties of a rational relation, e.g., reflexivity, transitivity, symmetry, turn out to be undecidable (with an input the transducer recognizing the relation), by reduction from the Post Correspondence Problem (PCP). Independently, Morvan \cite{morvan} has shown that the query $\exists x Rxx$ on rational frames is undecidable, as well. The reduction of PCP here is straightforward: given a PCP $\{(u_{1},v_{1}),\ldots ,(u_{n},v_{n})\}$, consider a transducer with only one state, which is both initial and accepting, and it allows the transitions $(u_{1},v_{1}),\ldots ,(u_{n},v_{n})$. Then, the PCP has a solution precisely if some pair $(w,w)$ is accepted by the transducer. Inclusion and equality of rational relations are undecidable, too,  \cite{Berstel}.

Furthermore, in \cite{ThomasConstructing} W. Thomas has constructed a single rational graph with undecidable first-order theory, by encoding the halting problem of a universal Turing machine.

\subsection{Rational Kripke models}

Rational graphs can be viewed as Kripke frames, hereafter called \emph{rational Kripke frames}.

\begin{definition}
A Kripke model $\mathcal{M}=\left(\mathcal{F},V\right)=\left(S,R,V\right)$ is a \textbf{rational Kripke model} (RKM) if the frame $\mathcal{F}$ is a rational Kripke frame, and the valuation $V$ assigns a regular language to each propositional variable, i.e., $V\left(p\right)\in\mathsf{REG}\left(\Sigma^*\right)$ for every $ p\in\Phi$. A
valuation satisfying this condition is called a rational valuation.
\end{definition}

\begin{example} \label{petrirat}
In  this example we will present a RKM based on the configuration graph of a Petri net. To make it self-contained, we give the basic relevant definitions here; for more detail see e.g., \cite{Reis1985}. A Petri net is a tuple $\left(P,T,F,M\right)$ where $P$ and $T$ are disjoint finite sets and their elements are called \emph{places} and \emph{transitions} respectively. $F:\left(P\times T\right)\cup\left(T\times P\right)\rightarrow\mathbb{N}$ is called a flow function and is such that if $F\left(x,y\right)>0$ then there is an arc from $x$ to $y$ and $F\left(x,y\right)$ is the multiplicity of that arc.  Each of the places contain a number of \emph{tokens} and a vector of integers $M\in\mathbb{N}^{|P|}$  is called a \emph{configuration}  (or, \emph{marking}) of the Petri net if the i$^{th}$ component of $M$ is equal to the number of tokens at the i$^{th}$ place in the Petri net. The \emph{configuration graph} of $\mathcal{N}$ has as vertices all possible configurations of $\mathcal{N}$ and the edges represent the possible transitions between configurations.

Now, let $\mathcal{N}=\left(P,T,F,M\right)$ be a Petri net, where $P=\left\{p_1,p_2\right\},\;
T=\left\{t\right\},\;F\left(p_1,t\right)=2,\;F\left(t,p_2\right)=3$ and
$M=\left(4,5\right)$.  Let $\mathcal{M}=\left(S,R,V\right)$ where $S=0^*10^*$, $R$ the transition relation of the configuration graph of $\mathcal{N}$ and $V$ the valuation defined by $V\left(p\right)=0010^*$ and $V\left(q\right)=0^*1000$.  Then $\mathcal{M}$ is a RKM and can be presented by the various machines in Figure \ref{finitepresent}.

\begin{figure}[here]
\caption{A finite presentation $\mathcal{M}$:
$A_1,A_2$ and $A_3$ recognize $S,V\left(p\right)$ and $V\left(q\right)$ respectively, and $T$ recognizes $R$.}
\label{finitepresent}

\vspace{15mm}
\begin{pspicture}(13,5)

\rput(0,5.5){\rnode{A}{$A_1:$}}
\pnode(0,4.5){p0}
\cnodeput(1,4.5){q1}{$q_1$}
\cnodeput(4,4.5){q2}{$q_2$}
\cnode(4,4.5){0.41}{2}
\ncline{->}{p0}{q1}
\nccircle{->}{q1}{0.5cm}
\nbput{$0$}
\nccircle{->}{2}{0.5cm}
\nbput{$0$}
\ncline{->}{q1}{2}
\naput{$1$}

\rput(5,5.5){\rnode{B}{$A_2:$}}
\pnode(5,4.5){pp0}
\cnodeput(6,4.5){qq1}{$p_1$}
\cnodeput(9,4.5){qq2}{$p_2$}
\cnode(9,4.5){0.41}{22}
\ncline{->}{pp0}{qq1}
\ncline{->}{qq1}{22}
\naput{$001$}
\nccircle{->}{22}{0.5cm}
\nbput{$0$}

\rput(0,3){\rnode{C}{$A_3:$}}
\pnode(0,2){pq0}
\cnodeput(1,2){qp1}{$r_1$}
\cnodeput(4,2){qp2}{$r_2$}
\cnode(4,2){0.41}{2q}
\ncline{->}{qp1}{2q}
\naput{$1000$}
\nccircle{->}{qp1}{0.5cm}
\nbput{$0$}
\ncline{->}{pq0}{qp1}

\rput(5,3){\rnode{D}{$T:$}}
\pnode(5,2){q0}
\cnodeput(6,2){p1}{$s_1$}
\cnodeput(9,2){p2}{$s_2$}
\cnodeput(12,2){p3}{$s_3$}
\cnode(12,2){0.41}{3}
\ncline{->}{q0}{p1}
\ncline{->}{p1}{p2}
\naput{$001/1$}
\ncline{->}{p2}{3}
\naput{$\epsilon/000$}
\nccircle{->}{p1}{0.5cm}
\nbput{$0/0$}
\nccircle{->}{p2}{0.5cm}
\nbput{$0/0$}

\end{pspicture}
\vspace{-15mm}
\end{figure}
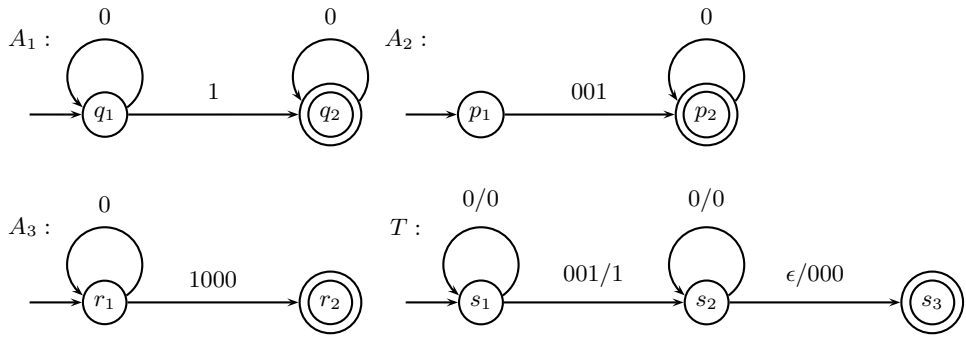

\end{example}

\section{Synchronized products of transducers and automata}
\label{sec:Synchronized}

In this section $\epsilon$ will denote the empty word, but will also be treated as a special symbol in an extended  alphabet.

\begin{definition}
Let $u$ be a word in some alphabet $\Gamma$ and $\gamma\in\Gamma$. The $\gamma$\textbf{-reduction} of $u$, denoted $u|_{\gamma}$, is the word obtained from $u$ after deleting all occurrences of $\gamma$. Likewise, if  $Y$ is a language in the alphabet $\Gamma$, then the $\gamma$\textbf{-reduction} of $Y$, denoted $Y|_{\gamma}$, is the language consisting of all $\gamma$-reductions of words in $Y$.
\end{definition}

\begin{lemma}
\label{lem:reduction}

If $Y$ is a regular language over an alphabet $\Gamma$ then $Y|_{\gamma}$ is a regular language over the alphabet $\Gamma - \{\gamma\}$.
\end{lemma}

\begin{proof}\textit{(Sketch)}
An automaton $\mathcal{A}|_{\gamma}$ recognizing $Y|_{\gamma}$, called here the $\gamma$\textbf{-reduction} of $\mathcal{A}$ can be constructed from an automaton $\mathcal{A}$ recognizing $Y$ as follows:

\begin{enumerate}
\item Remove all $\gamma$-transitions.

\item Add $(q,\gamma',q'')$ as a transition in $\mathcal{A}|_{\gamma}$ whenever     $(q,\gamma,q')$ and $(q',\gamma',q'')$ are transitions in  $\mathcal{A}$ and  $\gamma \neq \gamma'$.

\item Finally, define the accepting states of $\mathcal{A}|_{\gamma}$ as all accepting states of $\mathcal{A}$ plus those states $q$ such that
    $(q\stackrel{\gamma^*}{\rightarrow}q')$ in $\mathcal{A}$ and $q'$ is an accepting  state in $\mathcal{A}$.
    \end{enumerate} \hspace{\fill} $\triangleleft$
\end{proof}

\begin{definition}
A \textbf{run} of a finite automaton $\mathcal{A}=\left<Q,\Sigma,q^0,F,\delta\right>$ is
a sequence of states and transitions of $\mathcal{A}$:
$q_0\stackrel{x_1}{\rightarrow}q_1\stackrel{x_2}{\rightarrow}q_2\cdots\stackrel{x_n}{\rightarrow}q_{n}$,
such that $q_0=q^0$, $q_j\in Q, x_j\in \Sigma$, and $q_{j} \in
\delta\left(q_{j-1},x_j\right)$ for every $j = 1,2,\ldots,n$.

A run is \textbf{accepting} if it ends in an accepting state.

Run and accepting runs of transducers are defined likewise.
\end{definition}

\begin{definition}
A \textbf{stuttering run} of a finite automaton
$\mathcal{A}=\left<Q,\Sigma,q^0,F,\delta\right>$ is a sequence
$q_0\stackrel{x_1}{\rightarrow}q_1\stackrel{x_2}{\rightarrow}q_2\cdots\stackrel{x_n}{\rightarrow}q_{n}$,
such that $q_0=q^0$, $q_j\in Q$, and either $x_j\in \Sigma$ and $q_{j} \in
\delta\left(q_{j-1},x_j\right)$, or $x_j = \epsilon$ and $q_{j}  = q_{j-1}$ for every $j = 1,2,\ldots,n$.

Thus, a stuttering run of an automaton can be obtained by inserting
$\epsilon$-transitions from a state to itself into a run of that automaton. If the latter run is accepting, we declare the stuttering run to be an \textbf{accepting stuttering run}.

A \textbf{stuttering word} in an alphabet $\Sigma$ is any word in $\Sigma\cup\left\{\epsilon\right\}$.

The \textbf{stuttering language} of the automaton $\mathcal{A}$ is the set
$L^{\epsilon}(\mathcal{A})$ of all stuttering words whose $\epsilon$-reductions are recognized by $\mathcal{A}$; equivalently,  all stuttering words for which there is an accepting stuttering run of the automaton.
\end{definition}

\begin{definition}
Let $\mathcal{T}=\left<Q_\mathcal{T},\Sigma,q_\mathcal{T}^0,F_\mathcal{T},\rho_\mathcal{T}
\right>$ be a transducer, and let $\mathcal{A}$ be a (non-deterministic) finite automaton given by
$\mathcal{A}=\left<Q_\mathcal{A},\Sigma,q_\mathcal{A}^0,F_\mathcal{A},\delta_\mathcal{A}\right>$.

The \textbf{synchronized product} of $\mathcal{T}$ with $\mathcal{A}$ is the finite automaton:
\[\mathcal{T}\rightthreetimes\mathcal{A}=\left<Q_\mathcal{T}\times Q_\mathcal{A},\Sigma,
\left(q_\mathcal{T}^0,q_\mathcal{A}^0\right),F_\mathcal{T}\times F_\mathcal{A},
\delta_{\mathcal{T}\rightthreetimes\mathcal{A}}\right>\]
where $\delta_{\mathcal{T}\rightthreetimes\mathcal{A}} : \left(Q_\mathcal{T}\times
Q_\mathcal{A}\right)\times\left(\Sigma\cup\left\{\epsilon\right\}\right) \rightarrow \mathcal{P}(Q_\mathcal{T}\times Q_\mathcal{A})$ is such that, for any $p_\mathcal{T}^1,p_\mathcal{T}^2\in Q_\mathcal{T}$ and
$p_\mathcal{A}^1,p_\mathcal{A}^2\in Q_\mathcal{T}$ then
$\left(p_\mathcal{T}^2,p_\mathcal{A}^2\right) \in
\delta_{\mathcal{T}\rightthreetimes\mathcal{A}}\left(\left(p_\mathcal{T}^1,p_\mathcal{A}^1\right),x\right)$
if and only if
\begin{enumerate}
\item either there exists a $y\in\Sigma$ such that
    $\delta_\mathcal{A}\left(p_\mathcal{A}^1,y\right)=p_\mathcal{A}^2$ and
    $\left(p_\mathcal{T}^1,x,y,p_\mathcal{T}^2\right)\in\rho_\mathcal{T}$,
\item or $\left(p_\mathcal{T}^1,x,\epsilon,p_\mathcal{T}^2\right)\in\rho_\mathcal{T}$
    and $p_\mathcal{A}^1=p_\mathcal{A}^2$.
\end{enumerate}
 \end{definition}

Note that every run $R_{\mathcal{T}\rightthreetimes\mathcal{A}} =
(p_\mathcal{T}^0,p_\mathcal{A}^0)
\stackrel{u_1}{\rightarrow}(p_\mathcal{T}^1,p_\mathcal{A}^1)\stackrel{u_2}
{\rightarrow}\cdots\stackrel{u_n}{\rightarrow}(p_\mathcal{T}^{n},p_\mathcal{T}^{n})$ of the automaton $\mathcal{T}\rightthreetimes\mathcal{A}$ can be obtained from a pair:
\newline a run $R_\mathcal{T} =
p_\mathcal{T}^0\stackrel{\left(u_1/w_1\right)}{\rightarrow}p_\mathcal{T}^1
\stackrel{\left(u_2/w_2\right)}{\rightarrow}p_\mathcal{T}^2\cdots
\stackrel{\left(u_n/w_n\right)}{\rightarrow}p_\mathcal{T}^{n}$
in $\mathcal{T}$, \newline
and a stuttering run $R^s_\mathcal{A} =
p_\mathcal{A}^0\stackrel{w_1}{\rightarrow}p_\mathcal{A}^1\stackrel{w_2}{\rightarrow}p_\mathcal{A}^2\cdots
\stackrel{w_n}{\rightarrow}p_\mathcal{A}^{n}$ in $\mathcal{A}$, \newline by pairing the respective states $p_\mathcal{T}^j$ and $p_\mathcal{A}^j$ and removing the output symbol $w_j$ for every  $j = 1,2,\ldots,n$.

Let the reduction of $R^s_\mathcal{A}$ be the run $R_\mathcal{A} =
q_\mathcal{A}^0\stackrel{v_1}{\rightarrow}q_\mathcal{A}^1\stackrel{v_2}{\rightarrow}q_\mathcal{A}^2\cdots
\stackrel{v_m}{\rightarrow}q_\mathcal{A}^{m}$, with $m\leq n$. Then we say that the run
$R_{\mathcal{T}\rightthreetimes\mathcal{A}}$ is a {\bf synchronization of the runs} $R_\mathcal{T}$ and $R_\mathcal{A}$.

Note, that the synchronization of accepting runs of $\mathcal{T}$ and $\mathcal{A}$ is an accepting run of $R_{\mathcal{T}\rightthreetimes\mathcal{A}}$.
The following lemma is now immediate:

\begin{lemma}
\label{lem:synchprod} Let
$\mathcal{T}=\left<Q_\mathcal{T},\Sigma,q_\mathcal{T}^0,F_\mathcal{T},\rho_\mathcal{T}
\right>$ be a transducer recognizing the relation $R(\mathcal{T})$ and let
$\mathcal{A}=\left<Q_\mathcal{A},\Sigma,q_\mathcal{A}^0,F_\mathcal{A},\delta_\mathcal{A}\right>$
be a finite automaton recognizing the language $L(\mathcal{A})$. Then the language recognized by the synchronized product of $\mathcal{T}$ and $\mathcal{A}$ is
\[L(\mathcal{T}\rightthreetimes\mathcal{A}) = \{ u \mid \exists w \in L^{\epsilon}(\mathcal{A})
(u R(\mathcal{T}) w). \} \]
\end{lemma}

\section{Model checking of K$_t$ in rational Kripke models}
\label{sec:MCinRKM}

In this section we will establish decidability of the basic model checking problems for formulae of K$_t$ in rational Kripke models.

\begin{lemma} \label{ratprevreg}
Let $\Sigma$ be a finite non-empty alphabet, $X\subseteq\Sigma^*$
a regular subset, and let $R\subseteq\Sigma^*\times\Sigma^*$ be a rational relation.  Then the sets
\[\left<R\right>X= \{ u\in\Sigma^* | \exists v\in X (u
R v) \}\] and
\[\left<R^{-1}\right>X= \{ u\in\Sigma^*|\exists v\in X(v
R u)\}\] are regular subsets of $\Sigma^*$.
\end{lemma}

\begin{proof}
This claim essentially follows from results of Nivat (see \cite{Berstel}). However, using Lemmas \ref{lem:reduction} and \ref{lem:synchprod}, we give a constructive proof, which explicitly produces automata that recognize the resulting regular languages. Let $\mathcal{A}$ be a finite automaton recognizing $X$ and $\mathcal{T}$ be a transducer recognizing $R$. Then, the $\epsilon$-reduction of the synchronized product of $\mathcal{T}$ with $\mathcal{A}$ is an automaton recognizing $\left<R\right>X$; for $\left<R^{-1}\right>X$ we take instead of $\mathcal{T}$ the transducer for $R^{-1}$ obtained from $\mathcal{T}$ by swapping the input and output symbols in the transition relation\footnote{Note that, in general, the resulting automata need not be minimal, because they may have redundant states and transitions.}. \hspace{\fill} $\triangleleft$
\end{proof}

\begin{example}
Consider the automaton $\mathcal{A}$ and transducer $\mathcal{T}$ in Figure \ref{autex}. The language recognized by $\mathcal{A}$ is  $X=1^*\left(1+0^+\right)$ and the relation $R$ recognized by $\mathcal{T}$ is $R = \left\{\left(1^n0,10^n1\right)^m\left(1^k,10^k\right) \mid n,m,k\in\mathbb{N}\right\}
\cup \left\{\left(1^n0,10^n1\right)^m\left(01^k,11^k\right) \mid n,m,k\in\mathbb{N}\right\}$, where $X_1X_2$ denotes the component-wise concatenation of the relations $X_1$ and $X_2$, i.e., $X_1X_2 = \{(u_1u_2,v_1v_2) \mid (u_1,v_1) \in X_1, (u_2,v_2) \in X_2 \}$. For instance, if we take $n=1$, $m=2$ and $k=3$ we obtain that $(10,101)^2(1^3,10^3) = (1010111,1011011000)\in R$ (coming from the first set of the union) and \newline $(10,101)^2(01^3,11^3) = (10100111,1011011111)\in R$ (coming from the second set of that union).

Then, the synchronized product $\mathcal{T}\rightthreetimes\mathcal{A}$ is the finite automaton given in Figure \ref{resultaut} recognizing $\left<R\right>X=0^*+0^*1^+$.
Note that it can be simplified by removing redundant states and edges.

\begin{figure}[h]
\caption{The automaton $\mathcal{A}$ and the transducer $\mathcal{T}$.}
\label{autex}

\vspace{5mm}

\begin{pspicture}(5,4.5)
\rput(0,4){\rnode{A}{$\mathcal{A}:$}}
\pnode(0,2){p0}
\cnodeput(1,2){q1}{$p_1$}
\cnodeput(4,3){q2}{$p_2$}
\cnode(4,3){0.41}{2}
\cnodeput(4,1){q3}{$p_3$}
\cnode(4,1){0.41}{3}
\ncline{->}{p0}{q1}
\ncline{->}{q1}{2}
\naput{$0$}
\ncline{->}{q1}{3}
\naput{$1$}
\nccircle{->}{q1}{0.5cm}
\nbput{$1$}
\nccircle[angleA=270]{->}{2}{0.5cm}
\nbput{$0$}

\rput(6,4){\rnode{T}{$\mathcal{T}:$}}
\pnode(6,3){p0}
\cnodeput(7,3){q1}{$q_1$}
\cnodeput(10,3){q2}{$q_2$}
\cnode(10,3){0.41}{2}
\cnodeput(9,1){q3}{$q_3$}
\cnode(9,1){0.41}{3}
\ncline{->}{p0}{q1}

\nccurve[ncurv=.5,angleB=160,angleA=20]{->}{q1}{2}
\naput{$\epsilon/1$}
\nccurve[ncurv=.5,angleB=340,angleA=200]{->}{2}{q1}
\naput{$0/1$}
\ncline{->}{q1}{3}
\nbput{$0/1$}

\nccircle[angleA=270]{->}{2}{0.5cm}
\nbput{$1/0$}
\nccircle[angleA=270]{->}{3}{0.5cm}
\nbput{$1/1$}

\end{pspicture}
\vspace{-15mm}
\end{figure}
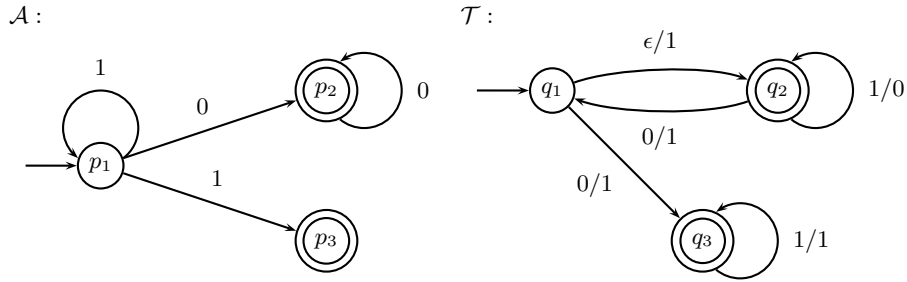

\bigskip

\begin{figure}[h]
\caption{The synchronized product $\mathcal{T}\rightthreetimes\mathcal{A}$ recognizing $\left<R\right>X$.}
\label{resultaut}
\begin{center}
\begin{pspicture}(11,5)

\rput(0,4.3){\rnode{TA}{$\mathcal{T}\rightthreetimes\mathcal{A}:$}}
\pnode(4,0){p0}
\cnodeput(4,1){q1}{$q_1$,$p_1$}
\cnodeput(8.5,1){q2}{$q_1,p_2$}
\cnodeput(10,3){q3}{$q_1,p_3$}
\cnodeput(4,3){q4}{$q_2,p_1$}
\cnodeput(7,3){q5}{$q_2,p_2$}
\cnode(7,3){0.42}{5}
\cnodeput(7,1){q6}{$q_2,p_3$}
\cnode(7,1){0.42}{6}
\cnodeput(1,3){q7}{$q_3$,$p_1$}
\cnodeput(10,1){q8}{$q_3,p_2$}
\cnode(10,1){0.42}{8}
\cnodeput(1,1){q9}{$q_3,p_3$}
\cnode(1,1){0.42}{9}

\ncline{->}{p0}{q1}
\ncline{->}{q1}{q6}
\naput{$\epsilon$}
\ncline{->}{q1}{q9}
\naput{$0$}
\ncline{->}{q1}{q7}
\naput{$0$}
\ncline{->}{q7}{q9}
\naput{$1$}
\ncline{->}{q4}{q5}
\nbput{$1$}
\nccurve[ncurv=.5,angleB=250,angleA=110]{->}{q1}{q4}
\naput{$\epsilon$}
\nccurve[ncurv=.5,angleB=70,angleA=290]{->}{q4}{q1}
\naput{$0$}
\nccurve[ncurv=.5,angleB=155,angleA=25]{->}{q4}{q3}
\naput{$0$}
\nccircle[angleA=250]{->}{q5}{0.55cm}
\nbput{$1$}
\nccircle[angleA=90]{->}{q7}{0.55cm}
\nbput{$1$}

\end{pspicture}
\end{center}
\end{figure}
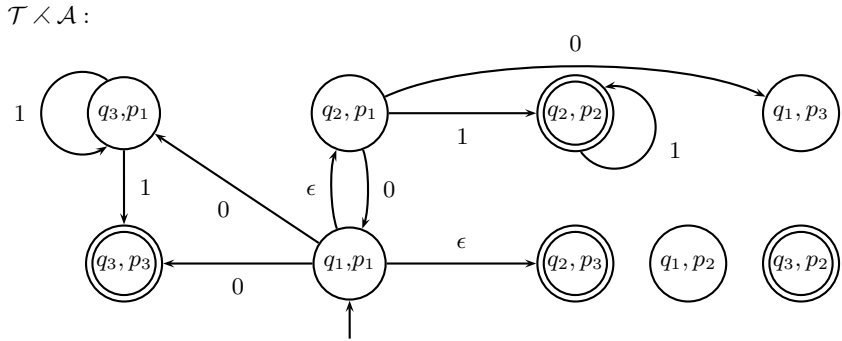

\end{example}

\begin{theorem} \label{modelcheckKt}
For every formula $\varphi \in$ {\bf K$_t$} and rational Kripke model $\mathcal{M}=\left(\Sigma^*,R,V\right)$, the set $[\![\varphi]\!]_\mathcal{M}$ is a rational language, effectively computable from $\varphi$ and the rational presentation of $\mathcal{M}$.
\end{theorem}

\begin{proof}
We prove the claim by induction on $\varphi$.

\begin{enumerate}
\item If $\varphi$ is an atomic proposition, the claim follows from the definition of a rational model.
\item The boolean cases follow from the effective closure of regular languages under boolean operations.
\item If $\varphi=\left<R\right>\psi$ then
    $[\![\varphi]\!]_{\mathcal{M}}=\left<R\right>[\![\psi]\!]_{\mathcal{M}}$, which is regular by the inductive hypothesis and Lemma \ref{ratprevreg}. Likewise for
    the case $\varphi=\left<R^{-1}\right>\psi$.
    \hspace{\fill} $\triangleleft$
\end{enumerate}
\end{proof}

\medskip
We now consider the following algorithmic model checking problems, where the  Kripke model is supposed to be given by some effective presentation:

\begin{enumerate}
\item \emph{Local model checking:} given a Kripke model $\mathcal{M}$, a state $s$ in $\mathcal{M}$, and a formula $\varphi$ of {\bf K$_t$}, determine whether $\mathcal{M},s\models\varphi$.
\item \emph{Global model checking:} given a Kripke model $\mathcal{M}$ and a formula $\varphi$ of {\bf K$_t$}, determine (effectively) the set $[\![\varphi]\!]_{\mathcal{M}}$ of all states in $\mathcal{M}$
where $\varphi$ is true.
\item \emph{Checking satisfiability in a model:} given a Kripke model $\mathcal{M}$ and a formula $\varphi$ of {\bf K$_t$}, determine whether $[\![\varphi]\!]_{\mathcal{M}} \neq \emptyset$.
\end{enumerate}

\begin{corollary}
\label{modelcheckCor}
Local model checking, global model checking, and checking satisfiability in a model, of formulae in {\bf K$_t$} in rational Kripke models are decidable.
\end{corollary}

\begin{proof}
Decidability of the global model checking follows immediately from Theorem \ref{modelcheckKt}. Then, decidability of the local model checking and of checking  satisfiability in a rational model follow respectively from the decidability of membership in a regular language, and of non-emptiness of a regular language (see e.g., \cite{Martin}). \hspace{\fill} $\triangleleft$
\end{proof}

\section{Complexity}
\label{sec:Complexity}

We will now attempt to analyze the complexity of global model checking a formula in {\bf K$_t$} on a rational Kripke model. Depending on which of these is fixed, we distinguish two complexity measures (see e.g., \cite{KuperVardi}): {\bf formula (expression) complexity} (when the model is fixed and the formula is feeded as input) and {\bf structure complexity} (when the formula is fixed and the model is feeded as input).

\subsection{Normal forms and ranks of formulae}

We will first need to define some standard technical notions.

A formula $\varphi\in$ {\bf K$_t$} is in {\bf negation normal form} if every occurrence of the negation immediately precedes a propositional variable.  Clearly every formula $\varphi\in$ {\bf K$_t$} is equivalent to a formula $\psi\in$ {\bf K$_t$} in negation normal form, of size linear in the size $\varphi$. For the remainder of this section, we will assume that a formula $\varphi$ we wish to model check is in a negation normal form.

The modal rank of a formula counts the greatest number of nested modalities in the formula, while the alternating box (resp., diamond) rank of a formula counts the greatest number of nested alternations of modalities with an outmost box (resp., diamond) in that formula. Formally:

\begin{definition}
The {\bf modal rank} for a formula $\varphi\in$ {\bf K$_t$}, denoted by
$mr\left(\varphi\right)$ is defined inductively as follows:
\begin{enumerate}
\item if $p$ is an atomic proposition, then $mr\left(p\right) = 0$ and $mr\left(\neg p\right)=0$;

\item
    $mr\left(\phi_1\vee\psi_2\right)= mr\left(\phi_1\land\psi_2\right) = \max\left\{mr\left(\psi_1\right),mr\left(\psi_2\right)\right\}$;

\item $mr\left(\vartriangle\psi\right) = mr\left(\psi\right)+1$ where
    $\vartriangle\in\left\{ \left[R\right], \left<R\right>, \left[R^{-1}\right], \left<R^{-1}\right> \right\}$.
\end{enumerate}
\end{definition}

\begin{definition}
The {\bf alternating box rank} and {\bf alternating diamond rank} of a formula
$\varphi\in$ {\bf K$_t$}, denoted respectively by $ar_{\Box}(\varphi)$ and $ar_{\Diamond}(\varphi)$, are defined by simultaneous induction as follows, where $\vartriangle\in \{\Box, \Diamond\}$:

\begin{enumerate}
\item if $p$ is an atomic proposition, then $ar_\vartriangle\left(p\right)= 0$ and $ar_\vartriangle\left(\neg p\right)= 0$;

\item $ar_\vartriangle\left(\psi_1\vee\psi_2\right) = ar_\vartriangle\left(\psi_1\land\psi_2\right) =
    \max\left\{ar_\vartriangle\left(\psi_1\right),ar_\vartriangle\left(\psi_2\right)\right\}$;

\item
    $ar_{\Diamond}\left(\left<R\right>\psi\right)=ar_{\Box}\left(\psi\right)+1$
    and
    $ar_{\Box}\left(\left<R\right>\psi\right)=ar_{\Box}\left(\psi\right)$.

Likewise for $ar_{\Diamond}\left( \left< R^{-1} \right> \psi \right)$ and $ar_{\Box}\left(\left<R^{-1}\right>\psi\right)$.

\item
    $ar_{\Box}\left(\left[R\right]\psi\right)=ar_{\Diamond}\left(\psi\right)+1$
    and
    $ar_{\Diamond}\left(\left[R\right]\psi\right)=ar_{\Diamond}\left(\psi\right)$.

Likewise for $ar_{\Diamond}\left(\left[R^{-1}\right]\psi\right)$ and $ar_{\Box}\left(\left[R^{-1}\right]\psi\right)$.

\end{enumerate}

Finally, the {\bf alternation rank} of $\varphi$, denoted $ar\left(\varphi\right)$ is defined to be \[ar\left(\varphi\right) =
\max\left\{ar_{\Box}\left(\varphi\right),ar_{\Diamond}\left(\varphi\right)\right\}.\]
\end{definition}

For instance,
$ar_{\Box}(\left[R\right](\left<R\right> \left[R\right] p \lor \left[R\right] \left[R^{-1}\right] \lnot q)) = 3$ and
$ar_{\Diamond}(\left[R\right](\left<R\right> \left[R\right] p \lor \left[R\right] \left[R^{-1}\right] \lnot q)) = 2$, hence $ar(\left[R\right](\left<R\right> \left[R\right] p \lor \left[R\right] \left[R^{-1}\right] \lnot q)) = 3$.

\subsection{Formula complexity}

We measure the size of a finite automaton or transducer $\mathcal{M}$ by the number of transition edges in it, denoted $|\mathcal{M}|$.

\begin{proposition} \label{diamondscomplexity}
If $\mathcal{A}$ is an automaton recognizing the regular language $X$ and $\mathcal{T}$ a transducer recognizing the rational relation $R$, then the time complexity of computing an automaton recognizing $\left<R\right>^mX$ is in $O(|\mathcal{T}|^m|\mathcal{A}|)$.
\end{proposition}

\begin{proof}
The size of the synchronized product $\mathcal{T}\rightthreetimes\mathcal{A}$ of $\mathcal{T}$ and $\mathcal{A}$ is bounded above by $|\mathcal{T}||\mathcal{A}|$ and it can be computed in time $O(|\mathcal{T}||\mathcal{A}|)$. The claim now follows by
iterating that procedure $m$ times. \hspace{\fill} $\triangleleft$
\end{proof}

However, we are going to show that the time complexity of computing an automaton recognizing $\left[R\right]X$ is far worse.

For a regular language $X$ recognized by an automaton $\mathcal{A}$, we define
$R_X=\left\{\left(u,\epsilon\right)|u\in X\right\}$.  A transducer $\mathcal{T}$ recognizing $R_X$ can be constructed from $\mathcal{A}$ by simply replacing every edge $\left(q,x,p\right)$ in $\mathcal{A}$ with the edge $\left(q,x,\epsilon,p\right)$.

\begin{lemma} \label{complement}
Let $X$ be a regular language.  Then the complementation $\overline{X}$ of $X$ equals $\left[R_X\right]\emptyset $.
\end{lemma}

\begin{proof}
Routine verification. \hspace{\fill} $\triangleleft$
\end{proof}

Consequently, computing $\left[R_X\right]\emptyset$ cannot be done in less than exponential time in the size of the (non-deterministic) automaton $\mathcal{A}$ for $X$. This result suggests the following conjecture.

\begin{conjecture} \label{conj}
The formula complexity of global model checking of a {\bf K$_t$}-formula is
non-elementary in terms of the alternating box rank of the formula.
\end{conjecture}

\subsection{Structure complexity}

Next we analyze the {\bf structure complexity}, i.e. the complexity of global model checking a fixed formula $\varphi\in$ {\bf K$_t$} on an input rational Kripke model. Here the input is assumed to be the transducer and automata presenting the model.

Fix a formula $\varphi\in$ {\bf K$_t$} in negation normal form, then for any input rational Kripke model $\mathcal{M}$ there is a fixed number of operations to perform on the input transducer and automata that can lead to subsequent exponential blowups of the size of the automaton computing $[\![\varphi]\!]_{\mathcal{M}}$. That number is bounded
by the modal rank $mr\left(\varphi\right)$ of the formula $\varphi$, and therefore the structure complexity is bounded above by an exponential tower of a height not exceeding that modal rank:

\large
$$2^{\cdots^{(mr\left(\varphi\right) \textrm{\ \tiny  times})^{\cdots^{2^{|\mathcal{T}||\mathcal{A}|}}}}}$$
\normalsize

However, using the alternation rank of $\varphi$ and Proposition \ref{diamondscomplexity}
we can do better.

\begin{proposition}
The structure complexity of global model checking for a fixed formula $\varphi\in$ {\bf K$_t$} on an input rational Kripke model $\mathcal{M}$, presented by the transducer and automata $\left\{\mathcal{T},\mathcal{A}_1,\ldots,\mathcal{A}_n\right\}$, is bounded above by
\large
$$2^{\cdots^{(ar\left(\varphi\right) \textrm{\ \tiny  times})^{\cdots^{2^{P\left(|\mathcal{T}|\right)}}}}}$$
\normalsize where $P\left(|\mathcal{T}|\right)$ is a polynomial in $|\mathcal{T}|$ with leading coefficient not greater that $n2^c$ where $c \leq \max \{|\mathcal{A}_i| \mid i=1,\ldots n\}$ and degree no greater than $mr\left(\varphi\right)$.

\end{proposition}
\begin{proof}
The number of steps in the computation of $[\![\varphi]\!]_{\mathcal{M}}$, following the structure of $\varphi$, that produce nested exponential blow-ups can be bounded by the alternation rank, since nesting of any number of diamonds does not cause an exponential blow-up, while nesting of any number of boxes can be reduced by double complementation to nesting of diamonds;  e.g., $\left[R\right](\left[R\right] \left[R\right] p \lor \left[R^{-1}\right] \lnot q)$ can be equivalently re-written as $\lnot \left<R\right>(\left<R\right> \left<R\right> \lnot p \land \left<R^{-1}\right> q)$. The initial
synchronized product construction (when a diamond or box is applied to a boolean formula) produces an automaton of size at most $2^c |\mathcal{T}|$, the number of nested product constructions is bounded above by $mr\left(\varphi\right)$, and each of these multiplies the size of the current automaton by $|\mathcal{T}|$. In the worst case, all alternations would take place after all product constructions, hence the upper bound.
\hspace{\fill} $\triangleleft$
\end{proof}

\section{Model checking extensions of {\bf K$_t$} on rational models}
 \label{sec:OtherExt}

\subsection{Model checking hybrid extensions of {\bf K$_t$}}

A major limitation of the basic modal language is its inability to refer explicitly to states in a Kripke model, although the modal semantics evaluates modal formulae at states. Hybrid logics provide a remedy for that problem. We will only introduce some basic hybrid logics of interest here; for more details consult e.g., \cite{BRV,blackburn00representation}.

The  \emph{basic hybrid tense logic} {\bf H$_t$} extends the basic tense logic {\bf K$_t$} with a set of new atomic symbols $\Theta$ called \emph{nominals} which syntactically form a second type of atomic formulae, which are evaluated in Kripke models in \emph{singleton sets} of states. The unique state in the valuation of a nominal is called its \emph{denotation}.  Thus, nominals can be used in {\bf H$_t$} to refer directly to states.

Here is the formal definition of the set of formulae of {\bf H$_t$}:
\[\varphi=p \mid i \mid  \neg\varphi  \mid \varphi\vee \phi \mid \left<R\right>\varphi \mid \left<R^{-1}\right>\varphi,\]
where $i\in\Theta$ and $p\in\Phi$.

The basic hybrid logic {\bf H$_t$} can be further extended to {\bf H$_t\left(@\right)$} by adding the \emph{satisfaction operator} $@$, where the formula $@_i\varphi$ means `\emph{$\varphi$ is true at the denotation of $i$}'. A more expressive extension of {\bf H$_t$} is {\bf H$_t(U)$} involving the \emph{universal modality} with semantics $\mathcal{M}, v \models [U] \varphi$ iff $\mathcal{M}, w \models \varphi$ for every $w \in \mathcal{M}$. The operator $@$ is definable in {\bf H$_t(U)$} by $@_i \varphi :=
[U](i \rightarrow \varphi)$. Moreover, {\bf H$_t$} can be extended with the more expressive \emph{difference modality} $\langle D \rangle$ (and its dual $[D]$), where $\mathcal{M},v\models \langle D \rangle\varphi$ iff there exists a $w\ne v$ such that $\mathcal{M},w\models\varphi$. Note that $[U]$ is definable in {\bf H$_t(D)$} by $[U]\varphi:=\varphi\wedge [D]\varphi$.

Yet another extension of {\bf H$_t\left(@\right)$} is {\bf
H$_t\left(@,\downarrow\right)$} which also involves \emph{state variables} and
\emph{binders} that bind these variables to states. Thus, in addition to {\bf
H$_t\left(@\right)$}, formulae also include $\downarrow \!\! x.\varphi$ for $x$ a state variable. For a formula $\varphi$ possibly containing free occurrences of a state variable $x$, and $w$ a state in a given model, let $\varphi\left[x\leftarrow i_w\right]$ denote the result of substitution of all free occurrences of $x$ by a nominal $i_w$ in $\varphi$, where $w$ is the denotation of $i_w$.  Then the semantics of $\downarrow \!\! x.\varphi$ is defined by: $\mathcal{M},w\models\downarrow  \!\! x.\varphi$ iff
$\mathcal{M},w\models\varphi\left[x\leftarrow i_w\right]$.

\begin{proposition}
For  every formula $\varphi$ of the hybrid language {\bf H$_t(D)$} (and therefore, of {\bf H$_t\left(@\right)$} and of {\bf H$_t\left(U\right)$}) and every rational Kripke model $\mathcal{M}$, the set $[\![\varphi]\!]_\mathcal{M}$ is an effectively computable rational language.
\end{proposition}

\begin{proof}
The claim follows from Theorem \ref{modelcheckKt} since the valuations of nominals, being singletons, are rational sets, and the difference relation $D$ is a rational relation. The latter can be shown by explicitly constructing a transducer recognizing $D$ in a given rational set, or by noting that it is the complement of the automatic relation of equality, hence it is automatic itself, as the family of automatic relations is closed under complements (see e.g., \cite{KhoussainovNerode} or \cite{BlumensathGraedel}). \hspace{\fill} $\triangleleft$
\end{proof}

\begin{corollary}
Global  and local model checking, as well as satisfiability checking, of formulae of the hybrid language {\bf H$_t(D)$} (and therefore, of {\bf H$_t\left(@\right)$} and {\bf H$_t\left(U\right)$}, too) in rational Kripke models are decidable.
\end{corollary}

\begin{proposition}
Model  checking of the {\bf H$_t\left(@,\downarrow\right)$}-formula $\downarrow \!\! x.\left<R\right> x$ in {\bf H$_t\left(@,\downarrow\right)$} on a given input rational Kripke model is not decidable.
\end{proposition}

\begin{proof}
Immediate consequence from Morvan's earlier mentioned reduction \cite{morvan} of the model checking of $\exists x Rxx$ to the Post Correspondence Problem.
\hspace{\fill} $\triangleleft$
\end{proof}

\begin{proposition}
There is a rational Kripke model on which model checking formulae from the hybrid language is undecidable.
\end{proposition}

\begin{proof} \textit{(Sketch)}
The rational graph constructed by Thomas \cite{ThomasConstructing} can be used to prove this undecidability, since the first-order properties queried there are also expressible in {\bf H$_t\left(@,\downarrow\right)$}.
\hspace{\fill} $\triangleleft$
\end{proof}

\subsection{Counting modalities}

We now consider extensions of {\bf K$_t$} with counting (or, graded) modalities:

\begin{itemize}

\item $\Diamond^{\geq k}\varphi$ with semantics: `there exist at least $k$ successors where $\varphi$ holds';

\item $\Diamond^{\leq k}\varphi$ with semantics: `there exist at most $k$ successors where $\varphi$ holds';

\item $\Diamond^{k}\varphi$ with semantics: `there exist exactly $k$ successors where $\varphi$ holds';

\item $\Diamond^{\infty}\varphi$ with semantics: `there exist infinitely many successors where $\varphi$ holds'.

\end{itemize}

Clearly, some of these are inter-definable: $\Diamond^{k}\varphi := \Diamond^{\geq k}\varphi \land \Diamond^{\leq k}\varphi$, while $\Diamond^{\geq k}\varphi := \neg \Diamond^{\leq k-1}\varphi$ and $\Diamond^{\leq k}\varphi := \neg \Diamond^{\geq k+1}\varphi$.

We denote by {\bf C}$_t$ the extension of {\bf K}$_t$ with $\Diamond^{\infty}\varphi$ and all counting modalities for all integers $k \geq 0$. Further, we denote by {\bf C}$^0_t$ the fragment of {\bf C}$_t$ where no occurrence of a counting modality is in the scope of any modal operator.

\begin{proposition}
Local model checking of formulae in the language {\bf C}$^0_t$ in rational Kripke models is decidable.
\end{proposition}

\begin{proof}
First we note that each of the following problems: `\textit{Given an automaton $A$, does its language contain at most / at least / exactly $k$ / finitely / infinitely many  words?}' is decidable. Indeed, the case of finite (respectively infinite) language is well-known (see e.g., \cite{Martin}, pp. 186--189). A decision procedure\footnote{The procedure designed here is perhaps not the most efficient one. but, it will not make the complexity of the model checking worse, given the high overall complexity of the latter.} for recognizing if the language of a given automaton $\mathcal{A}$ contains at least $k$ words can be constructed recursively on $k$. When $k=1$ that boils down to checking non-emptiness of the language (\textit{ibid}). Suppose we have such a procedure $P_k$ for a given $k$. Then, a procedure for $k+1$ can be designed as follows: first, test the language $L(\mathcal{A})$ of the given automaton for non-emptiness by looking for any word recognized by it (by searching for a path from the initial state to any accepting state). If such a word $w$ is found, modify the current automaton to exclude (only) $w$ from its language, i.e. construct an automaton for the language $L(\mathcal{A}) \setminus \{w\}$, using the standard automata constructions. Then, apply the procedure $P_k$ to the resulting automaton.

Testing $L(\mathcal{A})$ for having at most $k$ words is reduced to testing for at least $k+1$ words; likewise, testing for exactly $k$ words is a combination of these.

Now, the claim follows from Theorem \ref{modelcheckKt}. Indeed, given a RKM $\mathcal{M}$ and a formula $\varphi \in$ {\bf C}$^0_t$, for every subformula $\Diamond^{c}\psi$ of $\varphi$, where $\Diamond^{c}$ is any of the counting modalities listed above, the subformula $\psi$ is in {\bf K}$_t$, and therefore an automaton for the regular language $[\![\psi]\!]_\mathcal{M}$ is effectively computable, and hence the question whether $\Diamond^{c}\psi$ is true at the state where the local model checking is performed can be answered effectively. It remains to note that every formula of {\bf C}$^0_t$ is a boolean combination of subformulae $\Diamond^{c}\psi$ where $\psi \in$ {\bf K}$_t$.
\hspace{\fill} $\triangleleft$
\end{proof}

At present, we do not know whether any of the counting modalities preserves regularity in rational models, and respectively whether global model checking in rational models of either of these languages is decidable.

\subsection{A presentation based extension}

Here we consider a `presentation-based' extension of the multi-modal version of {\bf K$_t$}, where the new modalities are defined in terms of word operations, so they only have meaning in Kripke models where the states are labeled by words (such as the rational Kripke models) hereafter called \emph{Kripke word-models}.

To begin with, for a given alphabet $\Sigma$, with every language $X\subseteq \Sigma^*$ we can uniformly associate the following binary relations in $\Sigma^*$:

\medskip

$X?:=\left\{\left(u,u\right)|u\in X \right\}$;

$\overrightarrow{X}:=\left\{\left(uv,v\right)|u\in X,v\in\Sigma^*\right\}$.

\begin{proposition}
For every regular language $X\subseteq \Sigma^*$ the relations $X?$ and
$\overrightarrow{X}$ are rational.
\end{proposition}

\begin{proof}
For each of these, there is a simple uniform construction that produces from the automaton recognizing $X$ a transducer recognizing the respective relation. For instance, the transducer for $\overrightarrow{X}$ is constructed as composition of the transducers (defined just like the composition of finite automata) for the rational relations  $\left\{\left(u,\varepsilon\right)\mid u\in X\right\}$ and
$\left\{\left(v,v\right)\mid v\in\Sigma^*\right\}$. The former is constructed from the automaton $\mathcal{A}$ for $X$ by converting every $a$-transition in
$\mathcal{A}$, for $a \in \Sigma$, to $(a,\varepsilon)$-transition, and the
latter is constructed from an automaton recognizing $\Sigma^*$ by converting every $a$-transition, for $a \in \Sigma$, to $(a,a)$-transition.
\hspace{\fill} $\triangleleft$
\end{proof}

This suggests a natural extension of (multi-modal) {\bf K$_t$} with an infinite family of new modalities associated with relations as above defined over the extensions of formulae. The result is a richer, PDL-like language which extends the star-free fragment of PDL with test and converse by additional program constructions corresponding to the regularity preserving operations defined above. We call that language `\emph{word-based star-free PDL (with test and converse)}', hereafter denoted \textbf{WPDL}.

Formally, \textbf{WPDL} has two syntactic categories, viz., \emph{programs} \textsf{PROG} and \emph{formulae} \textsf{FOR}, defined over given alphabet $\Sigma$, set of atomic propositions \textsf{AP}, and set of atomic programs (relations) \textsf{REL}, by mutual induction as follows:

Formulae \textsf{FOR}:
\[ \varphi ::= p \mid l_a \mid \neg \varphi \mid \varphi_1 \vee \varphi_2 \mid \langle \alpha
\rangle \varphi
\]
for $p \in\ $\textsf{AP}, $a \in \Sigma$, and $\alpha \in\ $\textsf{PROG}, where for each $a \in \Sigma$ we have added a special new atomic proposition $l_a$, used further to translate extended star-free regular expressions to
WPDL-formulae.

Programs \textsf{PROG}:
\[ \alpha ::= \pi \mid \alpha\prime \mid \alpha_1 \cup \alpha_2 \mid \alpha_1 \circ \alpha_2
 \mid \varphi? \mid \overrightarrow{\varphi}
\]
where $\pi\in\ $\textsf{REL} and $\varphi\in\ $\textsf{FOR}.

We note that \textbf{WPDL} is not a purely logical language, as it does not have semantics on abstract models but only on word-models (including rational Kripke models), defined as follows. Let $\mathcal{M} = (S,\{R_\pi\}_{\pi \in \textsf{REL}},V)$ be a Kripke word-model over an alphabet $\Sigma$, with a set of states $S\subseteq \Sigma^*$, a family of basic relations indexed with \textsf{REL}, and a valuation $V$ of the atomic propositions from \textsf{AP}. Then every formula $\varphi\in$ \textsf{FOR} is associated with the language $[\![\varphi]\!]_{\mathcal{M}}\subseteq\Sigma^*$, defined as before, where $[\![p]\!]_{\mathcal{M}} := V(p)$ for every $p\in \textsf{AP}$ and $[\![l_a]\!] := \{ a \} \cap S$ for every $a \in \Sigma$. Respectively, every program $\alpha$ is associated with a binary relation $R_{\alpha}$ in $\Sigma^*$, defined inductively as follows (where $\circ$ is composition of relations):

\begin{itemize}
\item $R_{\alpha\prime} := R^{-1}_{\alpha}$,

\item $R_{\alpha_1 \cup \alpha_2} := R_{\alpha_1} \cup R_{\alpha_2}$,

\item $R_{\alpha_1 \circ \alpha_2} := R_{\alpha_1} \circ R_{\alpha_2}$,

\item $R_{\varphi?} := [\![\varphi]\!]?$,

\item $R_{\overrightarrow{\varphi}} := \overrightarrow{[\![\varphi]\!]}$.

\end{itemize}

\begin{lemma} \label{lem:WPDL}
For every \textbf{WPDL}-formulae $\varphi,\psi$ and a Kripke word-model $\mathcal{M}$:
\begin{enumerate}
\item $[\![\langle \varphi? \rangle \psi]\!]_{\mathcal{M}} =
    [\![\varphi]\!]_{\mathcal{M}} \cap [\![\psi]\!]_{\mathcal{M}}$.

\item  $[\![\langle \overrightarrow{\varphi} \rangle \psi]\!]_{\mathcal{M}} =
    [\![\varphi]\!]_{\mathcal{M}};[\![\psi]\!]_{\mathcal{M}}$ (where ; denotes
    concatenation of languages).
\end{enumerate}
\end{lemma}

\begin{proof} Routine verification:
\begin{enumerate}
\item $[\![\langle \varphi? \rangle \psi]\!]_{\mathcal{M}} =  \{ w \in \Sigma^* \mid
    w R_{\varphi?} v$ for some $v \in [\![\psi]\!]_{\mathcal{M}} \}$

    $=  \{ w \in \Sigma^* \mid w = v$ for some $v \in [\![\varphi]\!]_{\mathcal{M}}$
    and $v \in [\![\psi]\!]_{\mathcal{M}} \} = [\![\varphi]\!]_{\mathcal{M}} \cap
    [\![\psi]\!]_{\mathcal{M}}$.

\item $[\![\langle \overrightarrow{\varphi} \rangle \psi]\!]_{\mathcal{M}} = \{ w \in
    \Sigma^* \mid w R_{\overrightarrow{\varphi}} v$ for some $v \in
    [\![\psi]\!]_{\mathcal{M}} \}$

    $= \{ uv \in \Sigma^* \mid u \in [\![\varphi]\!]_{\mathcal{M}}, v \in
    [\![\psi]\!]_{\mathcal{M}} \} =
    [\![\varphi]\!]_{\mathcal{M}};[\![\psi]\!]_{\mathcal{M}}$.
\end{enumerate}
\hspace{\fill} $\triangleleft$
\end{proof}

\begin{corollary}
For every \textbf{WPDL}-formula $\varphi$ and a rational Kripke model $\mathcal{M}$, the language $[\![\varphi]\!]_{\mathcal{M}}$ is an effectively computable from $\varphi$ regular language.

\end{corollary}

\begin{corollary}
Local and global model checking, as well as satisfiability checking, of \textbf{WPDL}-formulae in rational Kripke models is decidable.
\end{corollary}

\emph{Extended star-free regular expressions} over an alphabet $\Sigma$ are defined as follows:
\[ E:= a \mid \neg E  \mid  E_1 \cup E_2 \mid E_1 ; E_2, \]
where $a\in\Sigma$. Every such expression $E$ defines a regular language $L(E)$, where $\neg,\cup, ;$ denote respectively complementation, union, and concatenation of languages. The question whether two extended star-free regular expressions define the same language has been proved to have a non-elementary complexity in \cite{MeyerStockmeyer}.

Every extended star-free regular expression can be linearly translated to an
\textbf{WPDL}-formula:
\begin{itemize}
\item $\tau(a) := l_a$,
\item $\tau(\neg E) := \neg \tau(E)$,
\item $\tau(E_1 \cup E_2) := \tau(E_1) \lor \tau(E_2)$,
\item $\tau(E_1 ; E_2) := \langle\overrightarrow{\tau(E_1)}\rangle\tau(E_2)$.
\end{itemize}

\begin{lemma} \label{lem:extreg} Given an alphabet $\Sigma$, consider the rational Kripke model $\mathcal{M}^{\Sigma}$ with set of states $\Sigma^*$, over empty sets of basic relations and atomic propositions. Then, for every extended star-free regular expression $E$,
\[ L(E) = [\![\tau(E)]\!]_{\mathcal{M}^{\Sigma}}.\]
\end{lemma}

\begin{proof} Straightforward induction on $E$. The only non-obvious case $E = E_1 ; E_2$ follows from Lemma \ref{lem:WPDL}.
\hspace{\fill} $\triangleleft$
\end{proof}

Consequently, for any extended star-free regular expressions $E_1$ and $E_2$, we have that $L(E_1) = L(E_2)$ iff $[\![\tau(E_1)]\!]_{\mathcal{M}^{\Sigma}} =
[\![\tau(E_2)]\!]_{\mathcal{M}^{\Sigma}}$ iff $\mathcal{M}^{\Sigma} \models \tau(E_1) \leftrightarrow \tau(E_2)$. Thus, we obtain the following.

\begin{corollary}
Global model checking of \textbf{WPDL}-formulae in rational Kripke models has
non-elementary formula-complexity.
\end{corollary}

Remark: since the $\overrightarrow{\varphi}$-free fragment of \textbf{WPDL} is expressively equivalent to {\bf K$_t$}, a translation of bounded exponential blow-up from the family of extended star-free regular expressions to the latter fragment would prove Conjecture \ref{conj}.

\section{Concluding remarks}
\label{sec:concl}

We have introduced the class of rational Kripke models and shown that all formulae of the basic tense logic \textbf{K$_t$}, and various extensions of it, have effectively computable rational extensions in such models, and therefore global model checking and local model checking of such formulae on rational Kripke models are decidable, albeit probably with non-elementary formula complexity.

Since model checking reachability on such models is generally undecidable, an important direction for further research would be to identify natural large subclasses of rational Kripke models on which model checking of K$_t$ extended with the reachability modality $\left<R\right>^*$ is decidable. Some such cases, defined in terms of the presentation, are known, e.g., rational models with length-preserving or length-monotone transition relation \cite{morvan}; the problem of finding structurally defined large classes of rational models with decidable reachability is still essentially open.

Other important questions concern deciding bisimulation equivalence between rational Kripke models, as that would allow us to transfer model checking of any property definable in the modal mu-calculus from one to the other. These questions are studied in a follow-up to the present work.

\section*{Acknowledgements}

This research has been supported by the National Research Foundation of South Africa through a research grant and a student bursary. We wish to thank Arnaud Carayol, Balder ten Cate, Carlos Areces, Christophe Morvan, and St\'ephane Demri, for various useful comments and suggestions. We are also grateful to the anonymous referee for his/her careful reading of the submitted version and many remarks and corrections which have helped us improve the content and presentation of the paper.

\end{document}